\documentclass[english]{article}
\usepackage[T1]{fontenc}
\usepackage[latin9]{inputenc}
\usepackage[a4paper]{geometry}
\usepackage{color}
\usepackage{babel}
\usepackage{float}
\usepackage{rotfloat}
\usepackage{mathtools}
\usepackage{url}
\usepackage{amsmath}
\usepackage{amsthm}
\usepackage{amssymb}
\usepackage{graphicx}
\usepackage[authoryear]{natbib}
\usepackage[unicode=true,
 bookmarks=true,bookmarksnumbered=false,bookmarksopen=false,
 breaklinks=true,pdfborder={0 0 1},backref=false,colorlinks=true]
 {hyperref}
\hypersetup{
 citecolor=blue}

\makeatletter
\theoremstyle{plain}
\newtheorem{thm}{\protect\theoremname}
\theoremstyle{plain}
\newtheorem{prop}[thm]{\protect\propositionname}
\ifx\proof\undefined
\newenvironment{proof}[1][\protect\proofname]{\par
	\normalfont\topsep6\p@\@plus6\p@\relax
	\trivlist
	\itemindent\parindent
	\item[\hskip\labelsep\scshape #1]\ignorespaces
}{%
	\endtrivlist\@endpefalse
}
\providecommand{\proofname}{Proof}
\fi
\theoremstyle{plain}
\newtheorem{lem}[thm]{\protect\lemmaname}

\usepackage{amsmath}
\DeclareMathOperator*{\argmax}{arg\,max}

\makeatother

\providecommand{\lemmaname}{Lemma}
\providecommand{\propositionname}{Proposition}
\providecommand{\theoremname}{Theorem}

\begin{document}
\title{Dynamic time series clustering via volatility change-points}
\author{Nick Whiteley\\
School of Mathematics, University of Bristol\\
and the Alan Turing Institute }
\maketitle
\begin{abstract}
This note outlines a method for clustering time series based on a
statistical model in which volatility shifts at unobserved change-points.
The model accommodates some classical stylized features of returns
and its relation to GARCH is discussed. Clustering is performed using
a probability metric evaluated between posterior distributions of
the most recent change-point associated with each series. This implies
series are grouped together at a given time if there is evidence the
most recent shifts in their respective volatilities were coincident
or closely timed. The clustering method is dynamic, in that groupings
may be updated in an online manner as data arrive. Numerical results
are given analyzing daily returns of constituents of the S\&P 500.
\end{abstract}

\section{Introduction\label{sec:Introduction}}

The purpose of this note is to outline, contextualize and demonstrate
a method for clustering time series using a change-point model. The
emphasis is on conveying the ideas of the method rather than depth
or generality, although some possible extensions and research directions
are given at the end in section \ref{sec:Extensions-and-directions}.
A Python implementation in the form of a Jupyter notebook is available:
\url{https://github.com/nckwhiteley/volatility-change-points}.

\subsection{Time series clustering}

Time series clustering is typically a two-step procedure. The first
step is to specify a pairwise measure of dissimilarity between series.
An overview of several popular approaches is given in \citep[Sec. 2]{montero2014tsclust}.
To mention just a few examples, this dissimilarity could be derived
from fairly simple statistics, such as cross-correlation; could involve
solving an optimization problem to find a `best' match between each
pair of series, for instance using the Fréchet distance or Dynamic
Time Warping \citep{berndt1994using}; or could involve fitting a
some form of model to each of the series, then computing a distance
between the fitted parameter values \citep{corduas2008time,otranto2008clustering}
or forecast distributions \citep{alonso2006time,vilar2010non}.

The second step is to pass the dissimilarity measure to an algorithm
which determines associations between the series. Again to mention
just a few popular techniques, hierarchical methods such as agglomerative
clustering, see for example \citep[Sec. 25.5]{murphy2012machine}
for an overview, form clusters sequentially. Each datum starts in
its own cluster and pairs of clusters are merged step-by-step in accordance
with some linkage criterion which quantifies how between-cluster dissimilarity
is derived from between-series dissimilarity. Centroid-based techniques
such as $k$-means \citep{macqueen1967some} or its generalizations
beyond Euclidean distance to, e.g., Bregman divergences \citep{banerjee2005clustering}
or Wasserstein distances \citep{ye2017fast} choose a collection of
cluster centers to minimize the sum of within-cluster divergences/distances.
The computational cost of global minimization is usually prohibitive
and so for implementation one settles for a local minimum obtained
using an iterative refinement method, such as Lloyd's algorithm \citep{lloyd1982least}
in the case of Euclidean distance.

A further level of sophistication is to approach clustering as a statistical
inference problem, with associations between data points and clusters
treated as latent variables to be inferred under a probabilistic model.
This allows uncertainty over clusterings, model parameters and model
structure to be quantified and reported in a principled manner. The
price to pay is usually an increased computational cost, for example
incurred through the EM algorithm, variational methods or Monte Carlo
sampling. The question of how to scale-up these methods to tackle
large data sets is an active topic of research. For a recent overview
and ideas involving parallelization and multi-step procedures, see
\citep{ni2018scalable}.

We propose a method which may be regarded as a half-way point between
a full-blown statistical treatment of time series clustering and the
simple two-step recipe described above. We do not perform probabilistic
modeling of associations, but we do perform probabilistic modeling
on a per-series basis and use it to define a notion of dissimilarity.

\subsection{Financial time series clustering}

Clustering of time series can serve a variety of purposes. In exploratory
data analysis one may simply want to discover groupings or unexpected
phenomena, and then summarize or report them for purposes of dimension
reduction or interpretation. Clustering may be one ingredient within
a broader statistical workflow, in which actions or decisions are
taken on the basis of discovered clusters.

Stemming from an influential paper of \citet{mantegna1999hierarchical},
clustering of financial time series using dissimilarity measures derived
from correlation has been applied to assist fundamental understanding
of markets, risk management, portfolio optimization and trading. A
comprehensive overview of research on this topic across machine learning,
econophysics, statistical physics, econometrics and behavioral finance
is maintained on arXiv by \citet{marti2017review}. The current version
includes a bibliography of over 400 references which we shall not
attempt to summarize.

An alternative approach to time series clustering, which does not
feature in \citep{marti2017review}, is to define dissimilarity by
some distance between parameter vectors obtained by fitting a model
to each of the series individually. \citet{otranto2008clustering}
gives a detailed account of dissimilarity measures in this vein, and
uses Wald tests and autoregressive metrics to measure the distance
between GARCH processes and thus cluster based on the heteroskedastic
characteristics. \citet{otranto2010identifying} extends this technique
to clustering based on distance between fitted Dynamic Conditional
Correlation models, and deploys the resulting covariance matrix estimates
within portfolio optimization.

As discussed by \citet{marti2016clustering} and \citet[Sec 4]{marti2017review},
a research topic still in its infancy but of considerable interest
is how to track changes in market structure, by recognizing clusters
which may change over time. Indeed \citet[Sec 4]{marti2017review}
report that many empirical studies do not achieve this but just deliver
a static clustering based all data available for a given time period.
An obvious step towards dynamic clustering is to apply a static clustering
method on a sliding window. If, for example, the clustering techniques
of \citet{otranto2008clustering,otranto2010identifying} were applied
in this manner, the length of the window would achieve a trade-off
between temporal locality and noisy parameter estimates, hence noisy
estimates of dissimilarity. The question of how long the window should
be in order to best deal with time-varying clusters is often not an
easy one to answer rigorously.

The method introduced below defines dissimilarity between time series
not in terms of correlation or parameter estimates, but rather in
terms of evidence about times of volatility change-points. As a consequence,
time series which evidence shifts in volatility around the same points
in times tend to be clustered together. This is of interest because
synchrony in volatility change-points across series may arise from
common underlying market factors or similar responses to changing
market conditions which the method may help to uncover. One appealing
feature of the method is that it naturally accommodates dynamic clustering,
in the sense that clusters can be re-evaluated at each point in time
as new data arrive, but it circumvents the need to work on a sliding
window: the underlying change-point model effectively adapts to the
time-scale of volatility changes of each series.

\section{The change-point model and dissimilarity measure}

\subsection{A generic change-point model for a single time series\label{subsec:A-generic-change-point}}

Consider a sequence of unobserved, integer valued and strictly increasing
change-points $(T_{n})_{n\in\mathbb{N}_{0}}$. $T_{0}$ is equal to
zero with probability one, and the increments $(T_{n}-T_{n-1})_{n\geq1}$
are i.i.d. with c.d.f. denoted by $G$.

For $t\in\mathbb{N}_{0}$, define $N(t)\coloneqq\sup\{n\geq0:T_{n}<t\}$,
$\tau_{t}\coloneqq T_{N(t)}$ and observe that $(\tau_{t})_{t\geq1}$
is a Markov chain, with transition probabilities:
\begin{equation}
p(\tau_{t+1}=s|\tau_{t}=u)=\begin{cases}
\frac{G(t-u)-G(t-1-u)}{1-G(t-1-u)}, & s=t,\\
\frac{1-G(t-u)}{1-G(t-1-u)}, & s=u\in\{0,\ldots,t-1\},\\
0, & \mathrm{otherwise},
\end{cases}\label{eq:tau_prior}
\end{equation}
corresponding to whether a new change-point has occurred or not.

Let $(y_{t})_{t\in\mathbb{N}_{0}}$ be observed returns which are
assumed to be jointly distributed with $(\tau_{t})_{t\geq1}$ such
that for each $t\geq1$,
\begin{equation}
p(\tau_{1:t},y_{0:t})=p(y_{0})\prod_{s=1}^{t}p(\tau_{s}|\tau_{s-1})p(y_{s}|\tau_{s},y_{0:s-1}),\label{eq:joint_dist}
\end{equation}
with the convention $p(\tau_{1}|\tau_{0})\equiv\delta_{0}(\tau_{1})$,
the Kronecker delta at $0$, to respect the fact that $T_{0}$ is
zero with probability $1$.

Consider the sequence of probability mass functions $(\pi_{t})_{t\geq1}$,
\begin{equation}
\pi_{t}(s)\coloneqq p(\tau_{t}=s|y_{0:t}).\label{eq:filtering_distribution}
\end{equation}
Again due to the fact that $T_{0}$ is zero with probability $1$,
we have $\pi_{1}(0)=1$. Combining the conditional independence structure
of (\ref{eq:joint_dist}) with (\ref{eq:tau_prior}), elementary marginalization
and Bayes' rule validate the following recursion, for $t\geq1$, 
\begin{equation}
\pi_{t+1}(s)\propto\begin{cases}
p(y_{t+1}|\tau_{t+1}=t,y_{0:t})\sum_{u=0}^{t-1}\left[\frac{G(t-u)-G(t-1-u)}{1-G(t-1-u)}\pi_{t}(u)\right], & s=t,\\
p(y_{t+1}|\tau_{t+1}=s,y_{0:t})\frac{1-G(t-s)}{1-G(t-1-s)}\pi_{t}(s), & s\in\{0,\ldots,t-1\},\\
0, & \mathrm{otherwise}.
\end{cases}\label{eq:recursion}
\end{equation}

This change-point model and the recursion (\ref{eq:recursion}) are
directly inspired by those of \citet{chopin2007dynamic}, \citet{fearnhead2007line}
and \citet{adams2007bayesian}. Our model is slightly more general
than those of \citet{fearnhead2007line} and \citet{adams2007bayesian},
who assumed that conditional on a change-point time, observations
after that time are independent of observations before. Note (\ref{eq:joint_dist})
does not imply such independence.

In section \ref{subsec:A-particular-instance} we describe an instance
of the above change-point model in which the terms $p(y_{t+1}|\tau_{t+1},y_{0:t})$
arise by analytically integrating out parameters associated with the
change-points under conjugate prior distributions. This makes our
setting more restrictive than that of \citet{chopin2007dynamic},
who did not assume such analytic integration possible, but instead
used numerical integration in the form of a sequential Monte Carlo
algorithm.

\subsection{The dissimilarity measure}

Suppose now that one is presented with $m\geq1$ series of returns
$\{(y_{t}^{i})_{t\in\in\mathbb{N}_{0}},i=1,\ldots,m\}$. Let $\pi_{t}^{i}$
be as in (\ref{eq:filtering_distribution}) with $(y_{t})_{t\in\in\mathbb{N}_{0}}$
there replaced by $(y_{t}^{i})_{t\in\in\mathbb{N}_{0}}$. We propose
to cluster the series at any given time $t$ with dissimilarity taken
to be a probability metric evaluated between the distributions $\{\pi_{t}^{i},i=1,\ldots,m\}$.
Thus series will be clustered at time $t$ if, through $\{\pi_{t}^{i},i=1,\ldots,m\}$,
they exhibit similar evidence about the times of their respective
most recent change-points.

Which probability metric to choose? It seems sensible to consider:
i) interpretation and ii) computational overhead. To address the first
of these two criteria, consider the total variation distance: 
\[
\mathrm{TV}(\pi^{1},\pi^{2})=\frac{1}{2}\sum_{s\in\mathbb{N}_{0}}|\pi^{1}(s)-\pi^{2}(s)|.
\]
The total variation distance is maximal and equal to $1$ as soon
as $\pi^{1}$ and $\pi^{2}$ have disjoint support, which is a rather
restrictive notion of dissimilarity. For instance, for two Kronecker
delta's $\pi^{1}=\delta_{t}$ and $\pi^{2}=\delta_{t+s}$,
\begin{equation}
\mathrm{TV}(\delta_{t},\delta_{t+s})=1,\quad\text{if}\quad|s|\neq0.\label{eq:tv}
\end{equation}

For describing distributions over change-points this insensitivity
to translation seems undesirable - a more appealing property might
be that the distance is strictly increasing in $|s|$. Alternatives
such as the Hellinger and $L_{2}$ distances involve similarly summing
of point-wise differences between probability mass functions, or functions
thereof, and hence have the same drawback. Divergences which involve
ratios of probability mass functions such as $\chi^{2}$ or Kullback-Liebler
similarly fail to express dissimilarity if the support of one mass
function is not contained within that of the other.

The total variation distance can be regarded as one instance of a
Wasserstein distance: given a distance $d(\cdot,\cdot)$ on $\mathbb{N}_{0}$
and $p\geq1$, the $p$'th Wasserstein distance associated with $d$
is:
\begin{equation}
W_{p}(\pi^{1},\pi^{2})\coloneqq\left(\inf_{\gamma\in\Gamma(\pi^{1},\pi^{2})}\sum_{(s,t)\in\mathbb{N}_{0}\times\mathbb{N}_{0}}d(s,t)^{p}\gamma(s,t)\right)^{1/p},\label{eq:wass_dist_defn}
\end{equation}
where $\Gamma$ is the set of all probability mass functions on $\mathbb{N}_{0}\times\mathbb{N}_{0}$
whose marginals are $\pi^{1}$ and $\pi^{2}$. The total variation
distance arises if one takes $d$ to be the discrete distance $d(s,t)=\mathbf{1}_{\{s\neq t\}}$
and $p=1$.

If instead $d$ is the usual distance on $\mathbb{N}_{0}$, $d(s,t)=|s-t|$,
we have
\begin{equation}
W_{p}(\delta_{t},\delta_{t+s})=|s|,\label{eq:wass_delta}
\end{equation}
and slightly more generally it can be shown by a direct computation
of the infimum in (\ref{eq:wass_dist_defn}) that:
\begin{equation}
W_{p}(a\delta_{s}+(1-a)\delta_{t},b\delta_{s}+(1-b)\delta_{t})=|a-b|^{1/p}|s-t|,\label{eq:wass_delta-1}
\end{equation}
see \citep[Ex 2.3]{bobkov2014one}. Whilst (\ref{eq:wass_delta})
and (\ref{eq:wass_delta-1}) are of course rather specific examples,
they illustrate the manner in which the Wasserstein distance associated
with $d(s,t)=|s-t|$ is more expressive than total variation distance
regarding translation, and therefore arguably more suited to our purposes
of comparing distributions over change-point times.

Turning to the criterion of computational overhead, in the case of
$d(s,t)=|s-t|$ on $\mathbb{N}_{0}$, the Wasserstein distance is
conveniently available in closed form:
\[
W_{p}(\pi^{1},\pi^{2})=\left(\int_{0}^{1}|F_{1}^{-1}(v)-F_{2}^{-1}(v)|^{p}\mathrm{d}v\right)^{1/p}
\]
where $F_{i}^{-1}(v)=\inf\{t\in\mathbb{N}_{0}:F_{i}(t)\geq v\}$ is
the generalized inverse c.d.f. of $\pi^{i}$. Even more conveniently
from a computational point of view, is the fact that:
\begin{equation}
W_{1}(\pi^{1},\pi^{2})=\sum_{s\in\mathbb{N}_{0}}|F_{1}(s)-F_{2}(s)|,\label{eq:wass_1}
\end{equation}
see \citep{bobkov2014one} for background.

With these considerations we shall settle on (\ref{eq:wass_1}) applied
to each pair $\pi_{t}^{i},\pi_{t}^{j}$ as our dissimilarity measure
at time $t$. Note that the support of any $\pi_{t}^{i}$ is always
contained in $\{0,1,\ldots,t-1\}$. Moreover, if the approximation
technique suggested in section \ref{subsec:Online-implementation}
is applied, each approximating distribution $\hat{\pi}_{t}^{i}$ (more
details later) will have a number of support points uniformly upper
bounded in t, and hence the cost of evaluating $W_{1}(\hat{\pi}_{t}^{i},\hat{\pi}_{t}^{j})$
is uniformly upper bounded in $t.$

Choosing the dissimilarity measure completes the first of the two
steps described in section \ref{sec:Introduction}. What options are
available for the second step? Hierarchical clustering can be performed
immediately after evaluating the pairwise distances and we shall illustrate
this approach through numerical experiments. For a centroid-based
approach in the style of $k$-means, one needs to introduce the notion
of Wasserstein barycentre, which is the Fréchet mean in the space
of probability distributions equipped with the Wasserstein distance.
Computing these barycentres is a non-trivial task in general, see
\citep{peyre2019computational} for numerical methods.

\subsection{Online implementation\label{subsec:Online-implementation}}

The number of terms in the summation in (\ref{eq:recursion}) clearly
increases linearly with $t$. Hence the cost of computing this recursion
from time zero up to some time $t$ is quadratic in $t$. A simple
route towards an online algorithm, i.e. one whose computational cost
per time step does not increase with time, is to introduce an approximation
to each $\pi_{t}$ with a number of support points uniformly upper
bounded in $t$.

For instance, consider a simple pruning strategy: fix a number of
support points $n\geq1$. For $t\leq n$, computer $\pi_{t}$ exactly
using the recursion (\ref{eq:recursion}). For $t>n$ assume one already
has an approximation to $\pi_{t}$, call it $\hat{\pi}_{t}$ which
has $n$ support points in $\{0,\ldots,t-1\}$. Then one can substitute
$\hat{\pi}_{t}$ for $\pi_{t}$ in the right hand side of (\ref{eq:recursion}),
and retain the $n$ of the $n+1$ resulting support points associated
with highest probabilities to give an approximation $\hat{\pi}_{t+1}$
to $\pi_{t+1}$.

A further consideration for online implementation of (\ref{eq:recursion})
is the cost of evaluating $p(y_{t+1}|\tau_{t+1},y_{0:t})$, does this
also increase with $t$? In the instance of the change-point model
described in section \ref{subsec:A-particular-instance}, we shall
show that $p(y_{t+1}|\tau_{t+1},y_{0:t})$ depends on $y_{0:t}$ through
statistics which can be updated online as data arrive, and hence it
is possible to evaluate the terms $p(y_{t+1}|\tau_{t+1},y_{0:t})$
sequentially in $t$ at a fixed cost per time step.

\subsection{A particular instance of the change-point model\label{subsec:A-particular-instance}}

Let $(T_{n})_{n\geq0}$ be distributed as in section \ref{subsec:A-generic-change-point}.
We now introduce a specific model for the returns $(y_{t})_{t\geq0}$
which, upon analytically marginalizing out certain parameters, will
satisfy (\ref{subsec:A-generic-change-point}) with a closed-form
expression for $p(y_{t+1}|\tau_{t+1},y_{0:t})$. In turn, this can
be plugged into the recursion (\ref{eq:recursion}) or its approximation
discussed in section \ref{subsec:Online-implementation}, in order
to evaluate the dissimilarity measure.

Consider a sequence of triples of parameters $(\mu_{n},\alpha_{n},\sigma_{n}^{2})_{n\in\mathbb{N}_{0}}$
and assume that:
\begin{equation}
y_{t}=\mu_{N(t)}+\alpha_{N(t)}y_{t-1}+\sigma_{N(t)}\epsilon_{t},\label{eq:r_t_eqn}
\end{equation}
where $(\epsilon_{t})_{t\geq1}$ are i.i.d. $\mathsf{N}(0,1)$. Thus
$(\mu_{n},\alpha_{n},\sigma_{n}^{2})$ parameterize the conditional
joint distribution of the data between the $n$'th and $(n+1)$'th
change-points, i.e. $(y_{T_{n}+1},\ldots,y_{T_{n+1}})$, given $y_{T_{n}}$.

We assume the following prior independence properties: the sequence
$(\mu_{n},\alpha_{n},\sigma_{n}^{2})_{n\geq0}$ and the sequence $(T_{n})_{n\geq0}$
are independent, and the triples $(\mu_{n},\alpha_{n},\sigma_{n}^{2})_{n\geq0}$
are independent across $n$. It can be shown that as a consequence
of these independences and (\ref{eq:r_t_eqn}),
\begin{align}
p(y_{t+1}|\tau_{t+1}=s,y_{0:t}) & =p(y_{t+1}|\tau_{t+1}=s,y_{s:t}),\label{eq:cond_ind_pred}\\
p(\mu_{N(t)},\alpha_{N(t)},\sigma_{N(t)}^{2}|\tau_{t}=s,y_{0:t}) & =p(\mu_{N(t)},\alpha_{N(t)},\sigma_{N(t)}^{2}|\tau_{t}=s,y_{s:t}).\label{eq:cond_ind_params}
\end{align}
The intuitive interpretation of these identities is that conditional
on the time of the most-recent change-point being $s$, data strictly
prior to $s$ are irrelevant to: predicting the next data point, as
per (\ref{eq:cond_ind_pred}), and inference for $\mu_{N(t)},\alpha_{N(t)},\sigma_{N(t)}^{2}$,
i.e., the parameters associated with the most recent change-point,
as per (\ref{eq:cond_ind_params}). 

To arrive at a closed-form expression for $p(y_{t+1}|\tau_{t+1}=s,y_{s:t})$
we set a zero-mean Normal-Inverse-Gamma prior distribution on each
parameter triple:

\begin{equation}
p(\mu_{n},\alpha_{n},\sigma_{n}^{2})=\frac{1}{2\pi|V_{0}|^{1/2}}\frac{b^{a}}{\Gamma(a)}\left(\frac{1}{\sigma_{n}^{2}}\right)^{a+2}\exp\left(-\frac{2b+\beta_{n}^{\mathrm{T}}V_{0}^{-1}\beta_{n}}{2\sigma_{n}^{2}}\right),\label{eq:prior_NIG}
\end{equation}
where $\beta_{n}\coloneqq[\mu_{n}\;\alpha_{n}]^{\mathrm{T}}$, $V_{0}\coloneqq\mathrm{diag(\delta_{0}^{2},\delta_{1}^{2})},$
and $a,b,\delta_{0},\delta_{1}$ are hyper-parameters which are common
across $n$.

The following proposition gives the expression for $p(y_{t+1}|\tau_{t+1}=s,y_{s:t})$
as desired, and marginal posterior densities for the parameters $\beta_{N(t)}$
and $\sigma_{N(t)}^{2}$ conditional on the time of the most recent
change-point.
\begin{prop}
\begin{align}
 & p(y_{t+1}|\tau_{t+1}=s,y_{s:t})=\mathsf{St}\left(2a_{s,t},h_{t+1}w_{s,t},\frac{b_{s,t}}{a_{s,t}}(1+h_{t+1}V_{s,t}h_{t+1}^{\mathrm{T}})\right),\label{eq:posterior_predictive}\\
 & p(\beta_{N(t)}|\tau_{t}=s,y_{s:t})=\mathsf{St}\left(2a_{s,t},w_{s,t},\frac{b_{s,t}}{a_{s,t}}V_{s,t}\right),\label{eq:parameter_posterior}\\
 & p(\sigma_{N(t)}^{2}|\tau_{t}=s,y_{s:t})=\mathsf{IG}(a_{s,t},b_{s,t}),\label{eq:sig^2_posterior}
\end{align}
where
\begin{align}
 & w_{s,t}\coloneqq V_{s,t}H_{s,t}^{T}y_{s+1:t},\label{eq:posterior_mean}\\
 & V_{s,t}\coloneqq(V_{0}^{-1}+H_{s,t}^{\mathrm{T}}H_{s,t})^{-1},\label{eq:posterior correlation}\\
 & a_{s,t}\coloneqq a+\frac{t-s}{2},\label{eq:posterior shape}\\
 & b_{s,t}\coloneqq b+\frac{1}{2}(\|y_{s+1:t}\|^{2}-w_{s,t}^{\mathrm{T}}V_{s,t}^{-1}w_{s,t}),\label{eq:posterior_scale}
\end{align}
$H_{s,t}\coloneqq[h_{t}^{\mathrm{T}}\cdots h_{s+1}^{\mathrm{T}}]^{\mathrm{T}}$
, $h_{t}\coloneqq[1\;y_{t-1}]$, and $y_{s+1:t}\equiv[y_{t}\;y_{t-1}\;\cdots\;y_{s+1}]^{\mathrm{T}}$.
\end{prop}
\begin{proof}
[Proof sketch.]Note from (\ref{eq:r_t_eqn}), 
\[
y_{\tau_{t}+1:t}=\left[\begin{array}{c}
y_{t}\\
\vdots\\
y_{\tau_{t}+1}
\end{array}\right]=H_{\tau_{t},t}\beta_{N(t)}+\sigma_{N(t)}\left[\begin{array}{c}
\epsilon_{t}\\
\vdots\\
\epsilon_{\tau_{t}+1}
\end{array}\right].
\]
The expressions in (\ref{eq:posterior_predictive})-(\ref{eq:sig^2_posterior})
can therefore be obtained by conditioning on $\tau_{t+1}=s$ or $\tau_{t}=s$,
and then applying standard results for Bayesian linear regression
under a Normal-Inverse-Gamma prior, see for example \citep[Sec 7.6.3]{murphy2012machine}.
\end{proof}
Further to the considerations in section \ref{subsec:Online-implementation}
it is important to notice that $(w_{s,t})_{t>s}$, $(V_{s,t})_{t>s}$,
$(a_{s,t})_{t>s}$, $(b_{s,t})_{t>s}$ can be calculated in a recursive
manner, so that the cost of evaluating each term of the form $p(y_{t+1}|\tau_{t+1},y_{0:t})$,
$p(y_{t+2}|\tau_{t+2},y_{0:t+1})$, etc. does not increase with $t$.
The following lemma gives the details.
\begin{lem}
For fixed $s\geq0$, 
\begin{align*}
V_{s,s+1} & =(V_{0}^{-1}+h_{s+1}^{\mathrm{T}}h_{s+1})^{-1}, & V_{s,t+1} & =V_{s,t}-\frac{V_{s,t}h_{t+1}^{\mathrm{T}}h_{t+1}V_{s,t}}{1+h_{t+1}V_{s,t}h_{t+1}^{\mathrm{T}}},\\
\tilde{y}_{s,s+1} & =y_{s+1}h_{s+1}^{\mathrm{T}}, & \tilde{y}_{s,t+1} & =\tilde{y}_{s,t}+y_{t+1}h_{t+1}^{\mathrm{T}},\\
\|y_{s+1:s+1}\|^{2} & =y_{s+1}^{2}, & \|y_{s+1:t+1}\|^{2} & =\|y_{s+1:t}\|^{2}+y_{t+1}^{2},\\
a_{s,s+1} & =a+\frac{1}{2}, & a_{s,t+1} & =a_{s,t}+\frac{1}{2},
\end{align*}
and for $t\geq s$,
\[
w_{s,t}=V_{s,t}\tilde{y}_{s,t}.
\]
\end{lem}
\begin{proof}
The expression for $V_{s,t+1}$ follows from 
\[
V_{s,t+1}^{-1}=V_{0}+H_{s,t+1}^{\mathrm{T}}H_{s,t+1}=V_{0}+H_{s,t}^{\mathrm{T}}H_{s,t}+h_{t+1}^{\mathrm{T}}h_{t+1}=V_{s,t}^{-1}+h_{t+1}^{\mathrm{T}}h_{t+1},
\]
 and the Sherman-Morrison formula. The other expressions are quite
straight forward.
\end{proof}

\subsection{Interpretation and relation to GARCH}

It is widely recognized that returns data at daily or higher frequencies
often exhibit certain stylized features: 
\begin{enumerate}
\item long-run mean or median close to zero, and heavy tails;
\item long-run auto-correlation of returns which is small or decays quickly
with lag-length, but auto-correlation of absolute or squared returns
which decays slowly;
\item time-dependent volatility. 
\end{enumerate}
To explain the interpretation of the change-point model in this context,
consider GARCH$(1,1)$:
\begin{align}
y_{t} & =\varsigma_{t}\epsilon_{t},\label{eq:GARCH_returns}\\
\varsigma_{t}^{2} & =c_{0}+c_{1}y_{t-1}^{2}+\rho\varsigma_{t-1}^{2},\label{eq:GARCH_variance}
\end{align}
where $(\epsilon_{t})_{t\geq0}$ is a white noise process. This is
perhaps the most widely used time series model which accommodates
the stylized features described above:
\begin{enumerate}
\item in the original presentation of \citet{bollerslev1986generalized},
$(\epsilon_{t})_{t\geq0}$ were taken as i.i.d. standard Gaussian,
so that the marginal distribution of $y_{t}$ under (\ref{eq:GARCH_returns})
is a scale-mixture of zero-mean Gaussians. To further account for
heavy-tails, \citet{bollerslev1987conditionally} suggested instead
a $t$-distribution centered at zero for $(\epsilon_{t})_{t\geq0}$,
with unit scale parameter;
\item due to the independence of the $(\epsilon_{t})_{t\geq0}$ and the
centering of their common distribution at zero, it is easily seen
that the autocorrelation of $(y_{t})_{t\geq0}$ (assuming it exists)
is zero. The sequence of squared returns $y_{t}^{2}$ from GARCH$(1,1)$
is an ARMA process \citep[Thm 7, p.61]{anderson2009handbook} and
hence may exhibit non-trivial autocorrelation;
\item time-dependent volatility is modelled through the `conditional-variance'
equation (\ref{eq:GARCH_variance}).
\end{enumerate}
These properties manifest themselves in the predictive distributions
$p(y_{t+1}|y_{0:t})$ associated with GARCH$(1,1)$; if indeed $(\epsilon_{t})_{t\geq0}$
are unit scale and zero-centered student's-$t$ variables with $2a$
degrees of freedom, then:
\begin{equation}
p(y_{t+1}|y_{0:t})=\mathsf{St}(2a,0,\varsigma_{t+1}^{2}),\label{eq:GARCH_predictive}
\end{equation}
where by writing out (\ref{eq:GARCH_variance}),
\begin{equation}
\varsigma_{t+1}^{2}=c_{0}\sum_{s=0}^{t}\rho^{s}+c_{1}\sum_{s=0}^{t}\rho^{s}y_{t-s}^{2}+\rho^{t+1}\varsigma_{0}^{2}.\label{eq:GARCH_parameter}
\end{equation}

Let us now explain the connection to (\ref{eq:posterior_predictive}).
For purposes of exposition, suppose that the parameters $(\mu_{n})_{n\geq0}$
are omitted from the change-point model, in the sense that (\ref{eq:r_t_eqn})
is simplified to:
\begin{equation}
y_{t}=\alpha_{N(t)}y_{t-1}+\sigma_{N(t)}\epsilon_{t},\label{eq:model_omit_mu}
\end{equation}
and suppose the prior on each parameter pair $(\alpha_{n},\sigma_{n}^{2}$)
is just the marginal prior of these two parameters under (\ref{eq:prior_NIG}). 
\begin{prop}
\label{prop:approx}Omitting $(\mu_{n})_{n\geq0}$ in the sense of
(\ref{eq:model_omit_mu}) results in the following expression for
(\ref{eq:posterior_predictive}):
\begin{equation}
p(y_{t+1}|\tau_{t+1}=s,y_{s:t})=\mathsf{St}(2a+t-s,\hat{\alpha}_{s,t}y_{t},\hat{\sigma}_{s,t}^{2}),\label{eq:approx_predictive_dist}
\end{equation}
where
\begin{align}
\hat{\alpha}_{s,t} & \coloneqq\frac{\sum_{i=s}^{t-1}y_{i}y_{i+1}}{\delta_{1}^{-1}+\sum_{i=s}^{t-1}y_{i}^{2}},\label{eq:approx_params_alpha}\\
\hat{\sigma}_{s,t}^{2} & \coloneqq\left[(1-\hat{\alpha}_{s,t}^{2})\frac{\sum_{i=s}^{t-1}y_{i}^{2}}{2a+t-s}+\frac{2b+y_{t}^{2}-y_{s}^{2}-\delta_{1}^{-1}}{2a+t-s}\right]\left(1+\frac{y_{t}^{2}}{\delta_{1}^{-1}+\sum_{i=s}^{t-1}y_{i}^{2}}\right)\label{eq:approx_params_sigma}
\end{align}
\end{prop}
Before giving the proof let us compare the predictive densities (\ref{eq:approx_predictive_dist})
and (\ref{eq:GARCH_predictive}).
\begin{itemize}
\item Consider the number of parameters in (\ref{eq:GARCH_parameter}) and
in (\ref{eq:approx_params_alpha})-(\ref{eq:approx_params_sigma}).
The former involves $a,c_{0},c_{1},\rho$ and $\varsigma_{0}^{2}$.
The latter involves $a,b,\delta_{1}$, but these parameters can effectively
be removed by considering the uninformative prior limits $\delta_{1}\to\infty$,
$a,b\to0$ , under which $p(y_{t+1}|\tau_{t+1}=s,y_{s:t})$ remains
well-defined as a probability density assuming $y_{i}\neq0$ for some
$i\in\{s,\ldots t-1\}$. By contrast, there appears not to be a prior
distribution under which one can analytically integrate out $a,c_{0},c_{1},\rho$
in GARCH$(1,1)$, so one must estimate these parameters or integrate
them out numerically, which would complicate the fitting of a change-point
model.
\item Concerning the stylized features of returns described above, the median
of $p(y_{t+1}|y_{0:t})$ in (\ref{eq:GARCH_predictive}) is clearly
zero. If $y_{s:t}$ exhibits little lag-one auto-correlation, in the
sense that $\hat{\alpha}_{s,t}\approx0$, then the median of (\ref{eq:approx_predictive_dist})
is approximately zero also. However, if this auto-correlation is non-zero,
this will be captured in (\ref{eq:approx_predictive_dist}), both
in terms of the centering at $\hat{\alpha}_{s,t}y_{t}$ and through
$\hat{\sigma}_{s,t}^{2}$. Thus the change-point model accommodates
but does no insist upon stylized feature 1) and zero autocorrelations
of returns in stylized feature 2); the model is flexible enough to
explain away variations in data which cannot be well modelled in terms
of dynamic volatility, such as short--lived trends and brief periods
of correlated returns. The squared scale parameter $\varsigma_{t+1}^{2}$
in (\ref{eq:GARCH_parameter}) is an exponentially-weighted average
of the previous squared returns $(y_{s}^{2})_{s\leq t}$. This is
what allows GARCH(1,1) to capture the auto-correlation of squared
returns as per stylized feature 2). The predictive distribution in
(\ref{eq:approx_predictive_dist}) achieves this in a slightly different
manner: $\hat{\sigma}_{s,t}^{2}$ involves a uniformly-weighted average
of the squared returns, $(y_{s}^{2},\ldots,y_{t-1}^{2})$, where $s$
is the time of the most recent change-point appearing in the conditioning
in (\ref{eq:approx_predictive_dist}). Thus the change-point model
can represent memory in the process of squared returns whilst avoiding
the need for the parameter $\rho$ in GARCH(1,1). Finally, Regarding
stylized feature 3), obviously the change-point model accommodates
changing volatility from one change-point to the next.
\item The degrees of freedom in (\ref{eq:GARCH_predictive}) is constant
at $2a$; in (\ref{eq:approx_predictive_dist}) the degrees of freedom
is $2a+t-s$, hence increasing as the time since the most recent change-point,
$t-s$, grows. As per (\ref{eq:model_omit_mu}), the change-point
model assumes volatility is constant between change-points and this
increase in the degrees of freedom reflects accumulation of data since
the most recent change-point, assuming it is known or we are conditioning
upon it. Integrating out the time of the most recent change-point
results in the following identities:
\begin{align*}
p(y_{t+1}|y_{0:t}) & =\sum_{s=0}^{t}p(y_{t+1}|\tau_{t+1}=s,y_{s:t})p(\tau_{t+1}=s|y_{0:t}),\\
p(\tau_{t+1}=s|y_{0:t}) & =\begin{cases}
\sum_{u=0}^{t-1}\frac{G(t-u)-G(t-1-u)}{1-G(t-1-u)}\pi_{t}(u), & s=t,\\
\frac{1-G(t-s)}{1-G(t-1-s)}\pi_{t}(s), & s\in\{0,\ldots,t-1\}.
\end{cases}
\end{align*}
Thus for the change-point model, the predictive density $p(y_{t+1}|y_{0:t})$
is a mixture of densities of the form (\ref{eq:approx_predictive_dist}),
i.e. of student's-$t$ distributions with varying degrees of freedom,
centering and scale parameters, where the mixing distribution is derived
from the posterior change-point distributions $\pi_{t}$. The parameter
posteriors $p(\beta_{N(t)}|y_{0:t})$ and $p(\sigma_{N(t)}^{2}|y_{0:t})$,
i.e. also with the time of the most recent change-point integrated
out, have similar mixture representations, the details are left to
the reader.
\item Re-introducing the parameter $(\mu_{n})_{n\geq0}$ in (\ref{eq:r_t_eqn})
allows non-zero median returns to be modelled, which may be desirable
over short periods or to accommodate short-lived market trends, but
is not accommodated in (\ref{eq:GARCH_returns})-(\ref{eq:GARCH_returns}).
Thus again the change-point model is flexible: if the data indicate
the median/mean is zero, as per stylized feature 1), or not, then
this will be reflected in the predictive distribution (\ref{eq:GARCH_predictive}).
\end{itemize}
In summary, the model described in section \ref{subsec:A-particular-instance}
has the convenient property that the parameters $(\mu_{n},\alpha_{n},\sigma_{n}^{2})_{n\in\mathbb{N}_{0}}$
can be integrated out analytically, thus allowing it to interface
with the generic change-point model and inference recursion in section
\ref{subsec:A-generic-change-point}. Its predictive distributions
are closely related to those of GARCH(1,1) and it accommodates the
standard stylized features of returns, but is flexible enough to also
model short-lived auto-correlations and trends.

\begin{proof}
[Proof of Proposition \ref{prop:approx}]Omitting $(\mu_{n})_{n\geq0}$
results in the simplifications: $\beta_{n}=\alpha_{n},$ $H_{s,t}=[y_{t-1}\,\cdots\,y_{s}]^{\mathrm{T}}$,
$h_{t}=y_{t-1}$, and $w_{s,t}$ and $V_{s,t}$ become scalars, in
particular:
\begin{align*}
 & w_{s,t}=V_{s,t}\sum_{i=s+1}^{t}y_{i}y_{i-1},\\
 & V_{s,t}=(\delta_{1}^{-1}+\sum_{i=s}^{t-1}y_{i}^{2})^{-1},\\
 & a_{s,t}=a+\frac{t-s}{2},\\
 & b_{s,t}=b+\frac{1}{2}\left[\sum_{i=s+1}^{t}y_{i}^{2}-\frac{\left(\sum_{i=s+1}^{t}y_{i}y_{i-1}\right)^{2}}{\delta_{1}^{-1}+\sum_{i=s}^{t-1}y_{i}^{2}}\right].
\end{align*}
Turning to the parameters of (\ref{eq:posterior_predictive}), we
find the simplifications:
\[
h_{t+1}w_{s,t}=y_{t}\frac{\sum_{i=s+1}^{t}y_{i}y_{i-1}}{\delta_{1}^{-1}+\sum_{i=s}^{t-1}y_{i}^{2}},
\]
\begin{align*}
\frac{b_{s,t}}{a_{s,t}}(1+h_{t+1}V_{s,t}h_{t+1}^{\mathrm{T}}y_{t+1}) & =\frac{b+\frac{1}{2}\left[\sum_{i=s+1}^{t}y_{i}^{2}-\frac{\left(\sum_{i=s+1}^{t}y_{i}y_{i-1}\right)^{2}}{\delta_{1}^{-1}+\sum_{i=s}^{t-1}y_{i}^{2}}\right]}{a+\frac{t-s}{2}}\left(1+\frac{y_{t}^{2}}{\delta_{1}^{-1}+\sum_{i=s}^{t-1}y_{i}^{2}}\right).
\end{align*}
A little rearranging completes the proof. 
\end{proof}

\section{Numerical results for constituents of the S\&P 500}

\subsection{Data and parameter settings}

All numerical experiments were based on a data set of daily prices
for stocks which were constituents of the S\&P500 index continuously
from 1998 to mid 2013. The data set was taken from \url{https://quantquote.com/historical-stock-data}.
According to source these data are split/dividend adjusted. All returns
referred to below are daily closing log returns, i.e. $y_{t}=\log(\text{price at \ensuremath{t}})-\log(\text{price at \ensuremath{t-1})}$.

When applying the change-point model from section \ref{subsec:A-generic-change-point},
the prior on each of the inter-change-point times, e.g., $T_{n}-T_{n-1}$,
was taken to be a geometric distribution shifted so its support is
$\{1,2,\ldots\}$ rather than $\{0,1,\ldots\}$. The parameter of
the geometric distribution was set to $0.02$. The hyper-parameters
in the prior distribution (\ref{eq:prior_NIG}) were taken to be $a=b=5\times10^{-4}$,
corresponding to a fairly uninformative prior over $\sigma_{n}^{2}$'s;
and $\delta_{0}=10$ and $\delta_{1}=0.02$, corresponding respectively
to an uninformative prior over the $\mu_{n}$'s and a prior over the
$\alpha_{n}$'s which places substantial mass on $[-1,1]$. The approximation
method described in section \ref{subsec:Online-implementation} was
implemented with the number of support points $n$ taken to be 100.

\subsection{Application of the change-point model to AMZN}

The objective of this section is to illustrate the output from the
change-point model applied to a single time series.

The top plot in figure \ref{fig:amzn params} shows the returns for
AMZN. The second plot shows the number of trading days since the maximum-a-posterior
(MAP) most recent change-point. To be precise, let $t$ be time since
the start of the data set on 1/1/1998 in units of trading days and
let $\tau_{t}^{\mathrm{MAP}}\coloneqq\argmax_{s}\pi_{t}(s)$. Then
the plot shows $t-\tau_{t}^{\mathrm{MAP}}$ against the calendar date
corresponding to $t$.

The third and fourth plots show means and 95\% credible regions for
$p(\mu_{N(t)}|\tau_{t}=\tau_{t}^{\mathrm{MAP}},y_{\tau_{t}^{\mathrm{MAP}}:t})$
and $p(\alpha_{N(t)}|\tau_{t}=\tau_{t}^{\mathrm{MAP}},y_{\tau_{t}^{\mathrm{MAP}}:t})$,
i.e. the two marginals of (\ref{eq:parameter_posterior}) with $\tau_{t}^{\mathrm{MAP}}$
plugged in. The interpretation of these distributions are that they
are the posterior distributions of the parameters associated with
the MAP most recent change-point. The bottom plot in figure \ref{fig:amzn params}
is constructed by finding the mode and 95\% credible region of $p(\sigma_{N(t)}^{2}|\tau_{t}=\tau_{t}^{\mathrm{MAP}},y_{\tau_{t}^{\mathrm{MAP}}:t})$,
i.e. (\ref{eq:sig^2_posterior}) with $\tau_{t}^{\mathrm{MAP}}$ plugged
in, and then mapping through $x\mapsto\frac{1}{2}\log x$, to give
the corresponding point estimate and credible region for $\log\sigma_{N(t)}$.

\begin{figure}[H]
\hfill{}\includegraphics[width=1\textwidth]{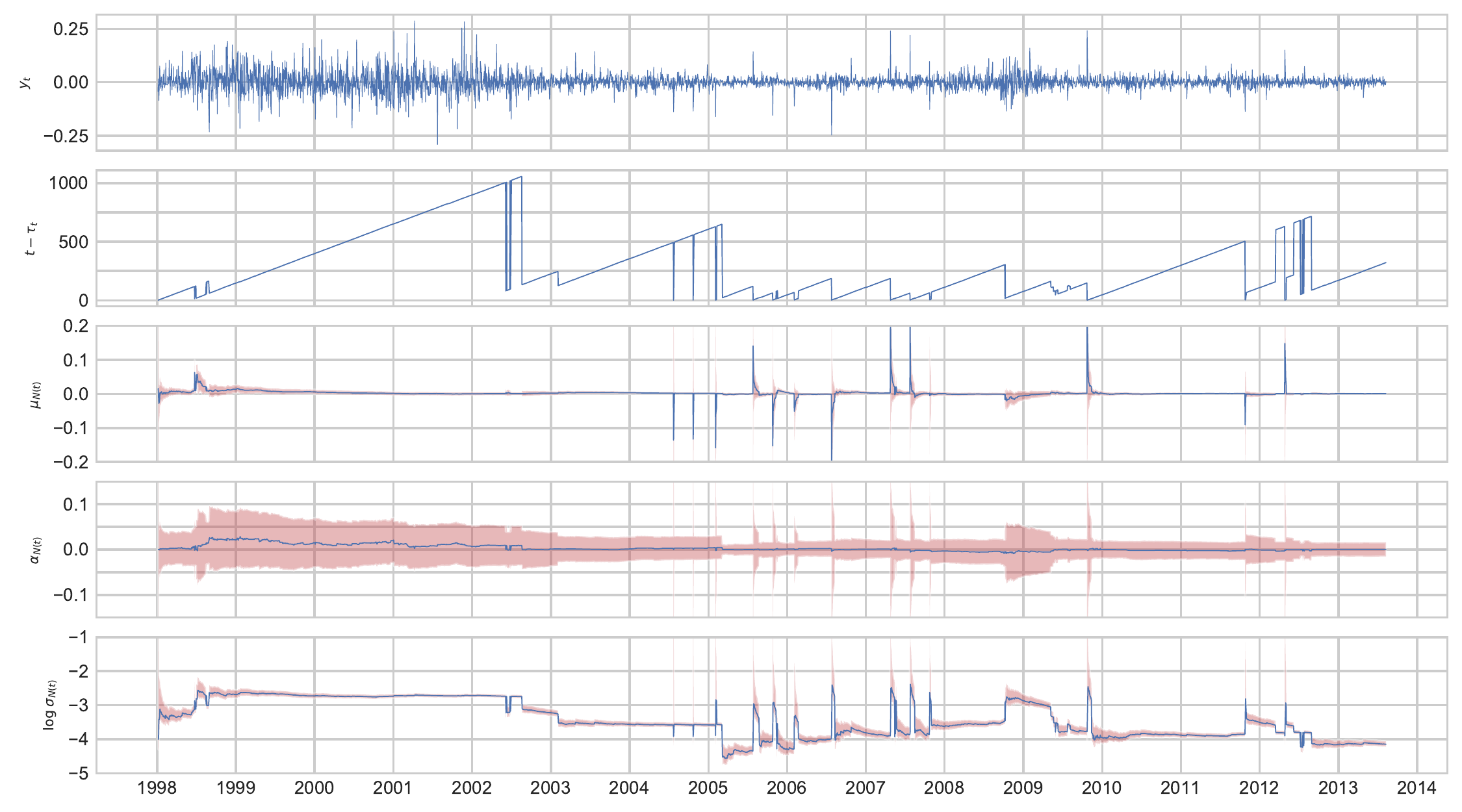}\hfill{}

\caption{\label{fig:amzn params}Change-point model applied to AMZN. From top
to bottom: adjusted daily closing log-returns; number of trading days
since MAP most recent change-point; posterior mean (blue) and 95\%
credible interval (red) for $\mu_{N(t)}$ conditional on MAP most
recent change-point; posterior mean (blue) and 95\% credible interval
(red) for $\alpha_{N(t)}$ conditional on MAP most recent change-point;
posterior mode (blue) and 95\% credible interval (red) for $\log\sigma_{N(t)}$.}
\end{figure}

To illustrate inference about change-point times beyond the simple
point estimate $\tau_{t}^{\mathrm{MAP}}\coloneqq\argmax_{s}\pi_{t}(s)$,
figure \ref{fig:amzn_cp} shows a snapshot of the returns from April
2007 until July 2009 and the change-point distributions $\pi_{t}$
for $t$ corresponding to 28/09/2008, 23/02/2009, 05/05/2009, 16/07/2009.
On 28/09/2008. i.e. just before the market crash, the change-point
distribution (second plot from top) shows a small amount of evidence
for a recent change, but most probability mass is associated with
24/07/2007 when the stock price surged after better-than-expected
Q2 results were announced. The third plot down, showing the change-point
distribution for $t$ corresponding to 23/02/2009 puts most of its
mass around the September 2008 market crash. In the fourth plot, corresponding
to 05/05/2009, the multiple modes in the distribution can be interpreted
as competing hypotheses about the most recent change point: the September
2008 market crash is amongst them, followed by crises in December
2008 and January-March 2009. The bottom plot picks up the change to
a period of lower volatility around the end of March 2009.

\noindent 
\begin{figure}[H]
\hfill{}\includegraphics[width=1\textwidth]{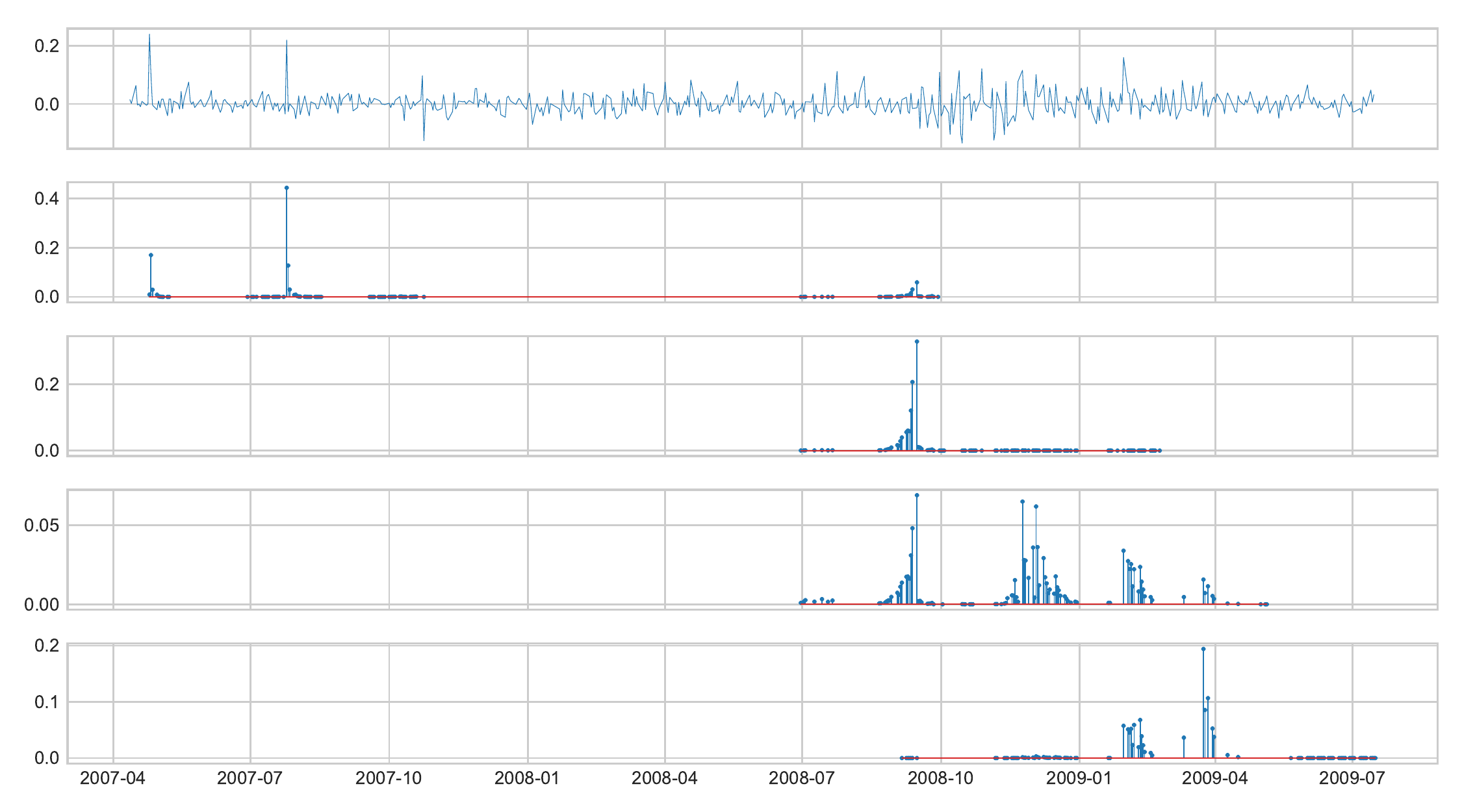}\hfill{}

\caption{\label{fig:amzn_cp}Change-point model applied to AMZN. Posterior
distributions over time of most recent change-point, $\pi_{t}$, for
$t$ corresponding to 29/09/2008, 23/02/2009, 05/05/2009, 16/07/2009.
Red lines on the horizontal axes indicate range of the support of
the distributions.}
\end{figure}

\noindent The top plot in figure \ref{fig:amzn_predictive} shows
the returns along with the 95\% credible region for each of the one-step-ahead
posterior predictive distributions $p(y_{t+1}|\tau_{t+1}=\tau_{t}^{\mathrm{MAP}},y_{\tau_{t}^{\mathrm{MAP}}:t})$,
i.e. (\ref{fig:amzn_predictive}) with $\tau_{t}^{\mathrm{MAP}}$
plugged in. The bottom plot shows these predictive credible regions
pushed forward to the price.
\begin{figure}[H]
\hfill{}\includegraphics[width=1\textwidth]{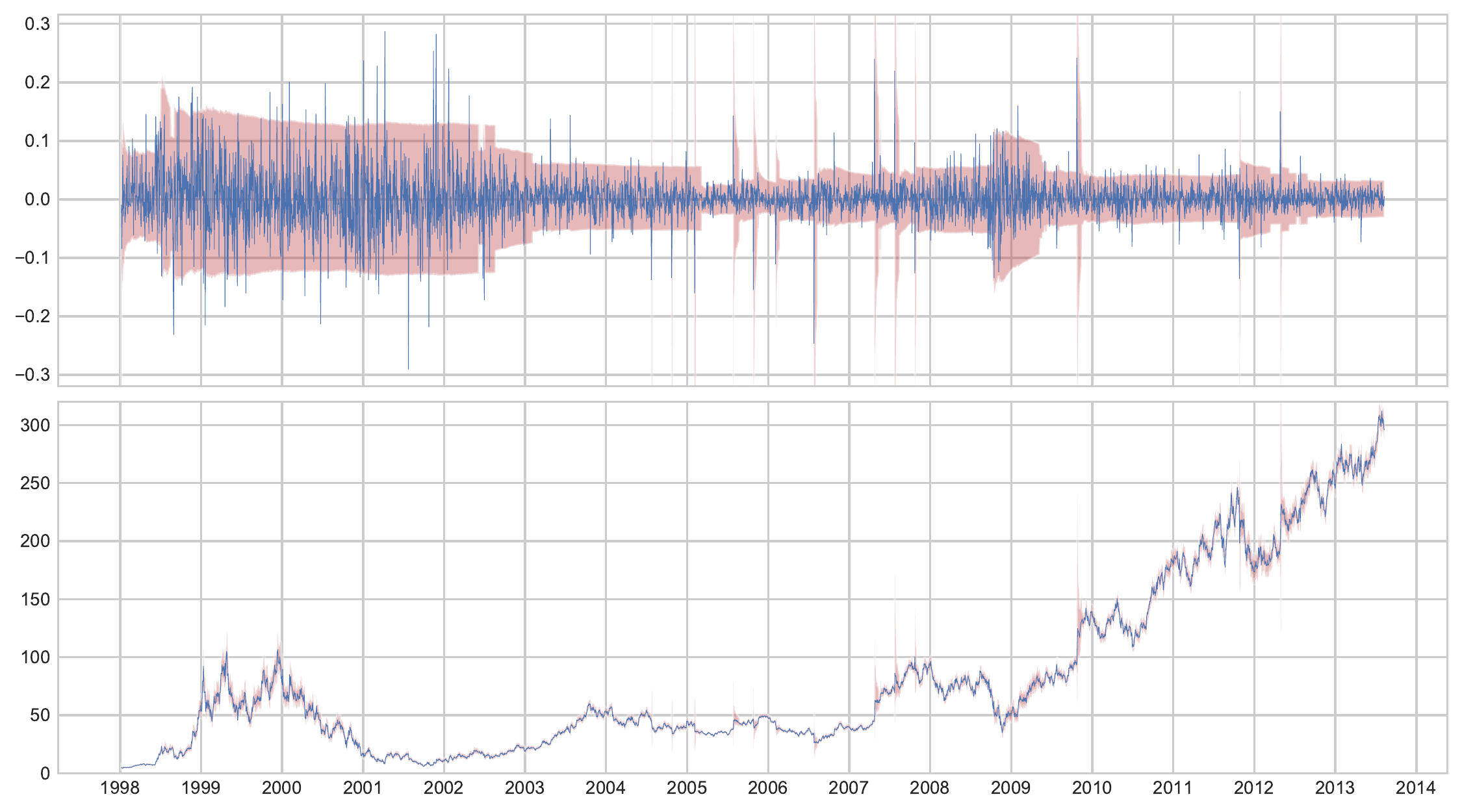}\hfill{}

\caption{\label{fig:amzn_predictive}Change-point model applied to AMZN. Blue
plot shows adjusted daily closing log-returns (top) and prices (bottom).
Red shading indicates posterior predictive 95\% credible interval
conditional on MAP most recent change-point. See text for definition.}
\end{figure}

\subsection{Hierarchical clustering}

Figure \ref{fig:Dissimilarity-matrix} shows the dissimilarity matrix
of Wassertein distances $W_{1}(\pi_{t}^{i},\pi_{t}^{j})$ as in (\ref{eq:wass_1})
across the first 80 S\&P 500 constituents by alphabetical order for
$t$ corresponding to 16/07/2009. The reason for considering only
80 constituents is to keep the following visual results simple and
easy to read. The date 16/07/2009 was chosen for purposes of illustration
as it post-dates the global financial crisis and the onset of the
subsequent recovery.

\noindent 
\begin{figure}[H]
\hfill{}\includegraphics[viewport=113.4bp 0bp 599.4bp 360bp,clip,width=1\textwidth]{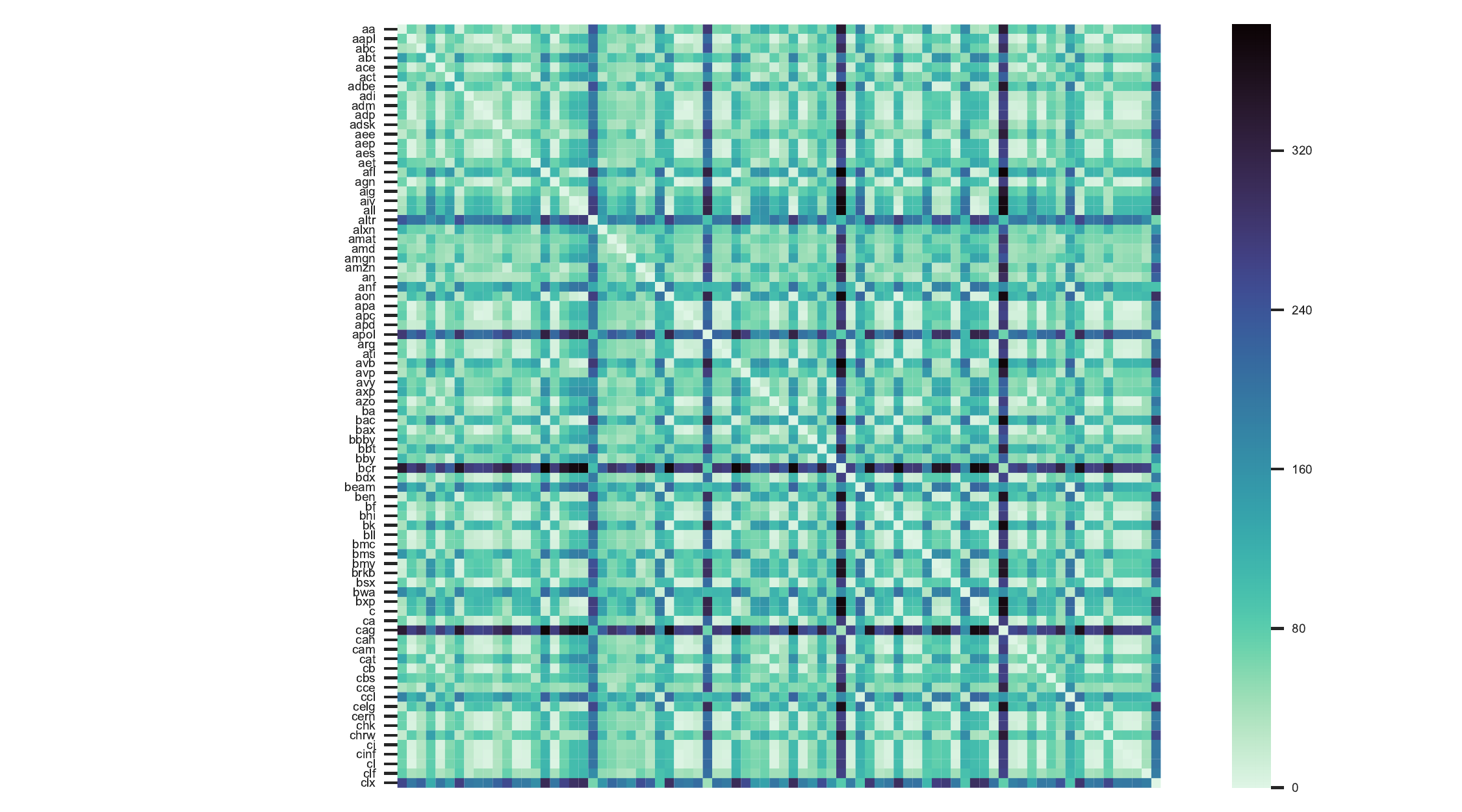}\hfill{}

\caption{\label{fig:Dissimilarity-matrix}Dissimilarity matrix for first 80
constituents of S\&P 500 for $t$ corresponding to 16/07/2009.}
\end{figure}

\noindent Whilst the dissimilarity matrix seems to show rich structure,
it is not easy to directly interpret. This is where hierarchical clustering
comes in: figure \ref{fig:Hierarchical-clustering-dendrogr} shows
the result of agglomerative clustering with the average linkage method,
implemented in Python using the Seaborn statistical data visualization
library, see \url{https://seaborn.pydata.org} and \url{https://SciPy.org}
for details of the underling linkage method.

This clustering method proceeds by initializing each stock in a separate
cluster, then sequentially combining nearby clusters and re-calculating
between-cluster distances. The output is a dendrogram, shown on the
right of figure \ref{fig:Hierarchical-clustering-dendrogr}, and a
re-ordering of the rows/columns of the dissimilarity matrix to respect
the structure of the dendrogram.

Once clusters of stocks are identified from this dendrogram, one may
then interrogate their respective change-point distributions. To illustrate
the idea, three clusters are highlighted in figure \ref{fig:Hierarchical-clustering-dendrogr}.
The first cluster consists of:
\begin{itemize}
\item AIV, Apartment Investment and Management, a real estate investment
trust;
\item AFL, AFLAC Incorporated, an insurance company;
\item AVB, AvalonBay Communities Real estate, an investment trust;
\item AON, Aon, an insurance broker, risk, retirement and health services
consulting company;
\item BK, Bank of New York;
\item BXP, Boston Properties, a real estate investment trust ;
\item ALL, Allstate, an insurance company;
\item BAC, Bank of America.
\end{itemize}
There is a clear theme of financial, real-estate and insurance sectors
to this cluster. Let us examine their change-point distributions $\pi_{t}$
for $t$ corresponding to 16/07/2009: they are shown in figure \ref{fig:2900 cluster 1}
and the common feature is probability mass around May 2009. It was
at this time that some stocks badly effected by the crisis in 2008
and early 2009 showed signs of recovery. Indeed inspecting the estimates
of $\sigma_{N(t)}^{2}$ for each of these stocks (not shown) reveals
there was a discrete in each of their volatilities around May 2009.

The second cluster is less sector-specific, consisting of:
\begin{itemize}
\item APA, Apache Corporation, a hydrocarbon exploration company;
\item AZO, AutoZone, an automotive parts retailer;
\item CHK, Chesapeake Energy, a hydrocarbon exploration company;
\item AAPL, Apple;
\item AGN, Allergan, a pharmaceutical company;
\item CL, Colgate-Palmolive, a consumer products company.
\end{itemize}
Inspecting figure \ref{fig:Hierarchical-clustering-dendrogr}, it
is clear that this cluster is one component of a larger cluster of
30+ stocks which have similar change-point distributions. Figure \ref{fig:2900 cluster 2}
indicates the feature they have in common is evidence of a volatility
change, or several volatility changes, in December 2008, but little
evidence of a change-point between then and June 2009. Broadly speaking,
these stocks were hit by the crisis in around October 2008, but their
volatility subsequently decreased sooner than the stocks in the first
cluster, around the end of 2008.

The third cluster is smaller, consisting of the three stocks:
\begin{itemize}
\item ABT, Abbott Laboratories, a healthcare company;
\item AXP, American Express;
\item CAT, Caterpillar, a construction equipment manufacturer.
\end{itemize}
Figure \ref{fig:2900 cluster 3} reveals that the feature these three
stocks have in common is evidence of a change-point around September
2008, about the time of the market crash, but little evidence of a
change-point between then and June 2009. In fact inspecting the estimated
volatility parameters $\sigma_{N(t)}^{2}$ for these stocks (not shown)
shows that all three of these stocks remained in a state of relatively
high volatility from September 2008 until around July 2009.

\noindent 
\begin{sidewaysfigure}[H]
\hfill{}\includegraphics[width=1\textwidth]{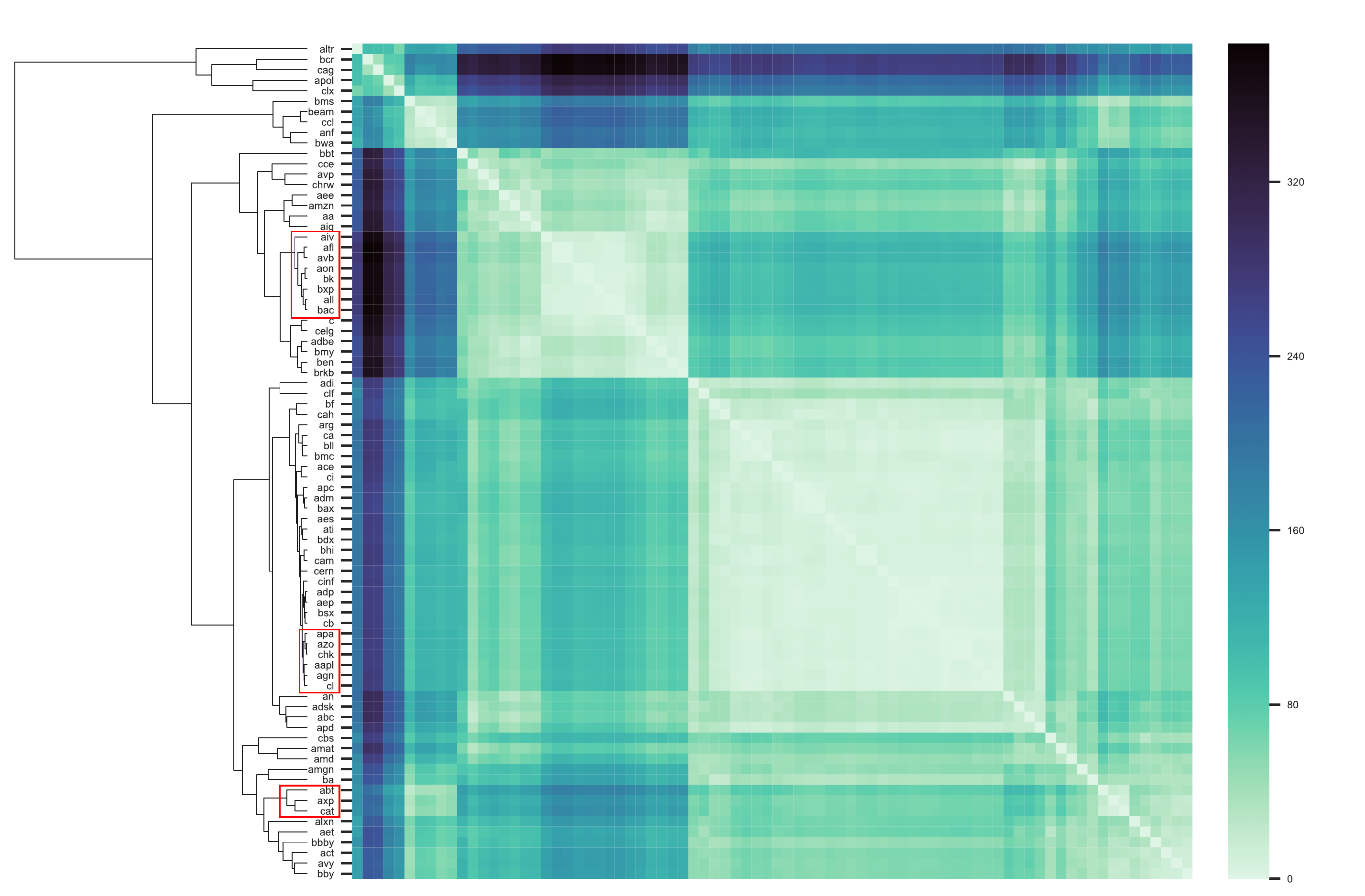}\hfill{}\caption{\label{fig:Hierarchical-clustering-dendrogr}Hierarchical clustering
dendrogram and re-ordered dissimilarity matrix for first 80 constituents
of S\&P 500 on 16/07/2009. Red highlighting of three clusters \{AIV,AFL,AVB,AON,BK,BXP,ALL,BAC\},
\{APA,AZO,CHK,AAPL,AGN,CL\}, \{ABT, AXP, CAT\}. The associated change-point
distributions are shown in figures \ref{fig:2900 cluster 1}-\ref{fig:2900 cluster 3}}
\end{sidewaysfigure}

\noindent 
\begin{figure}[H]
\hfill{}\includegraphics[width=1\textwidth]{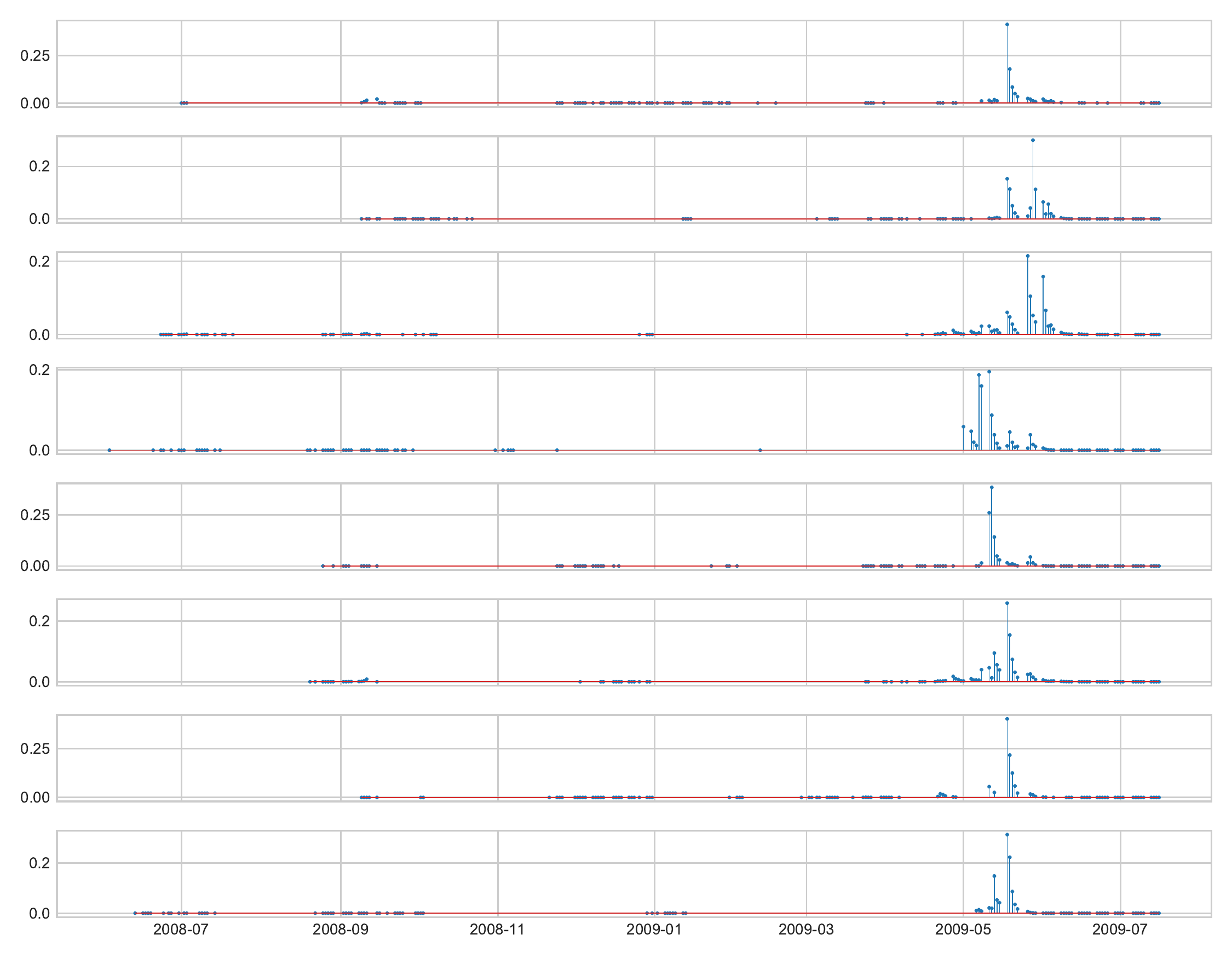}\hfill{}

\caption{\label{fig:2900 cluster 1}Posterior distributions of time of most
recent change-point as of 16/07/2009 for a cluster of S\&P500 constituents
extracted from the dendrogram in figure \ref{fig:Hierarchical-clustering-dendrogr},
from top to bottom: AIV, AFL, AVB, AON, BK, BXP, ALL, BAC.}
\end{figure}

\begin{figure}
\hfill{}\includegraphics[width=1\textwidth]{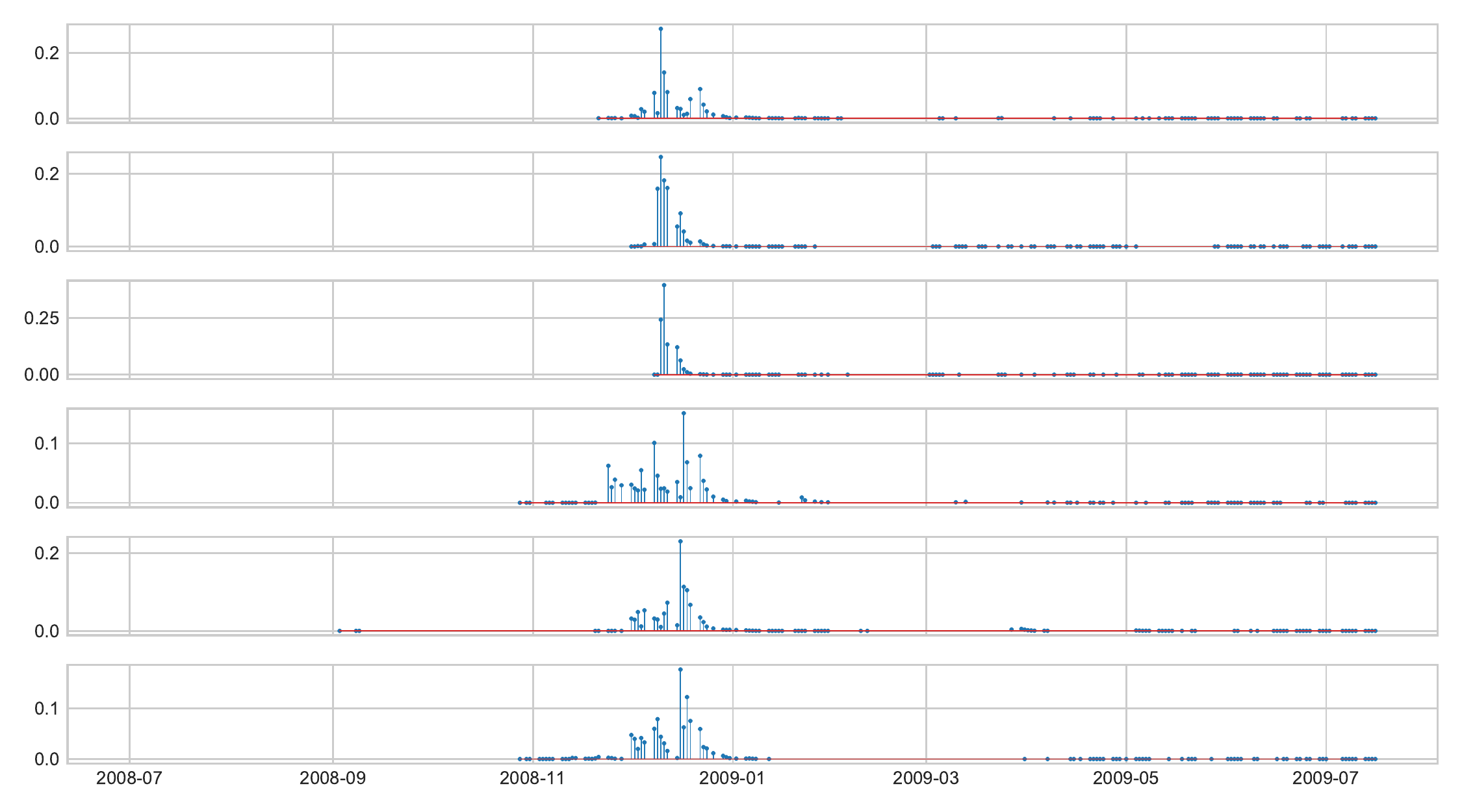}\hfill{}

\caption{\label{fig:2900 cluster 2}Posterior distributions of time of most
recent change-point as of 16/07/2009 for a cluster of S\&P500 constituents
extracted from the dendrogram in figure \ref{fig:Hierarchical-clustering-dendrogr},
from top to bottom: APA, AZO, CHK, AAPL, AGN, CL.}
\end{figure}
\begin{figure}
\hfill{}\includegraphics[width=1\textwidth]{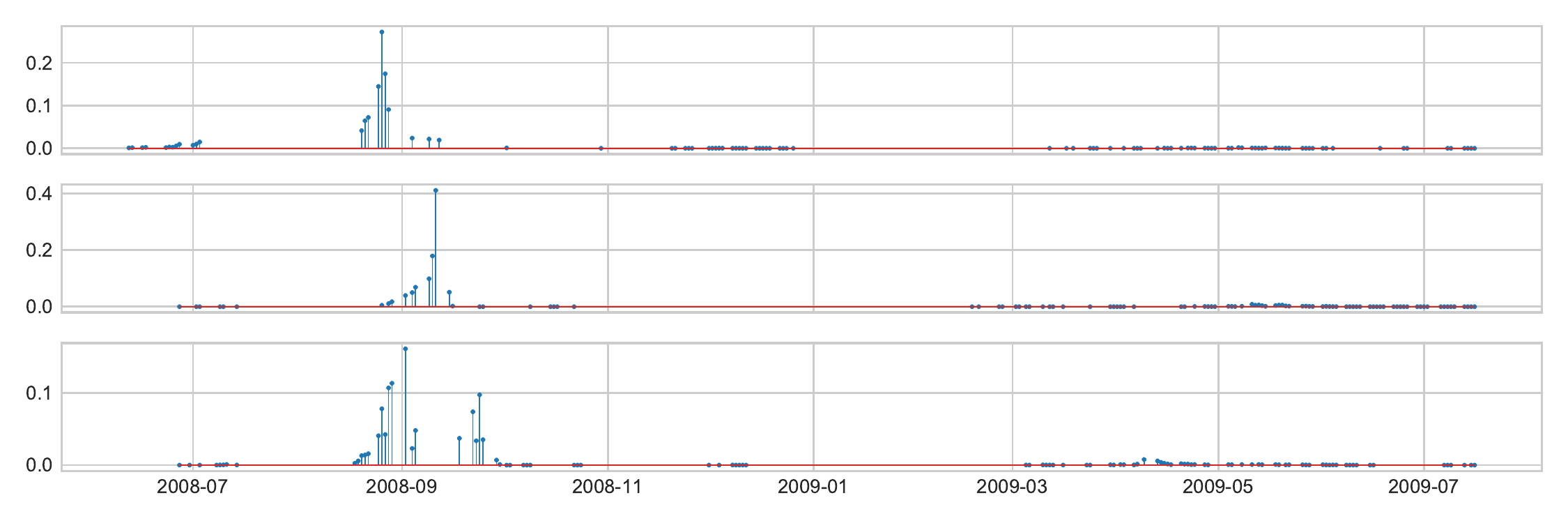}\hfill{}

\caption{\label{fig:2900 cluster 3}Posterior distributions of time of most
recent change-point as of 16/07/2009 for a cluster of S\&P500 constituents
extracted from the dendrogram in figure \ref{fig:Hierarchical-clustering-dendrogr},
from top to bottom: ABT, AXP, CAT.}
\end{figure}

\section{Extensions\label{sec:Extensions-and-directions}}

There are a number of avenues open for further investigation.

In terms of the modelling of individual time series, there are number
of ways the model from section \ref{subsec:A-particular-instance}
could be extended. As it stands, it doesn't explicitly model leverage
effects - that increases in volatility tend to be larger when recent
returns have been negative. A number of variants of the basic GARCH
model, such as Threshold-GARCH and Exponential-GARCH do model leverage
effects, but involve parameters for which conjugate priors are available.
It might be useful to find a half way point between such models and
that of section \ref{subsec:A-particular-instance}, to achieve more
accurate modelling, whilst retaining the analytic tractability which
allows parameters to be integrated out.

It could be desirable to develop a more principled approach to calibrating
the hyper-parameters $a,b$, $\delta_{0},\delta_{1}$, and the parameters
of the prior on the inter-change-point times. This could be approached,
for example, as a maximum likelihood or Bayesian inference problem.
Particle Markov Chain Monte Carlo methods for the latter are given
in \citep{whiteley2009}.

It could also be interesting to explore alternative probability metrics
and alternative clustering methods, for instance $k$-means using
Wasserstein Barycenters \citep{ye2017fast}.

\bibliographystyle{plainnat}
\bibliography{clustering}

\end{document}